%% file: dist_vms.tex
\documentclass[letterpaper, 10 pt, conference]{ieeeconf} 
\IEEEoverridecommandlockouts                              % This command is only
\usepackage{graphics} % for pdf, bitmapped graphics files
\usepackage{epsfig} % for postscript graphics files
\usepackage{amsmath} % assumes amsmath package installed
\usepackage{amssymb}  % assumes amsmath package installed
 \usepackage{subfigure}% subcaptions for subfigures
 \usepackage{subfigmat}% matrices of similar subfigures, aka small mulitples
 \usepackage{fancyvrb}%  extended verbatim environments
 \usepackage{lettrine}%  dropped capital letter at beginning of paragraph
 \usepackage[colorlinks]{hyperref}%  hyperlinks [must be loaded after dropping]
  \usepackage{amsmath,amssymb,url,listings,mathrsfs,wasysym}
  \usepackage{stmaryrd}
 \usepackage{algorithm}
\usepackage[noend]{algpseudocode}

\newtheorem{example}{Example}
\newtheorem{defn}{Definition}
\newtheorem{Theorem}{Theorem}
\newtheorem{Problem}{Problem}
\newtheorem{prop}{Proposition}

\newtheorem{rem}{Remark}

\def\beq{\begin{equation}}
\def\eeq{\end{equation}}

\newcommand{\Node}{\mathcal{V}}
\newcommand{\node}{v}
\newcommand{\Edge}{\mathcal{E}}
\newcommand{\Asmp}{A}
\newcommand{\Gar}{G}
\newcommand{\lra}{\lambda}
\newcommand{\asmp}[2] {A_{#1}^{#2}}
\newcommand{\gar}[2] {G_{#1}^{#2}}
\newcommand{\ws}[1] {\llbracket {#1} \rrbracket}
\title{SAT-based Distributed Reactive Control Protocol Synthesis for Boolean Networks}
\author{Yunus Emre Sahin, Necmiye Ozay% <-this % stops a space
	\thanks{A preliminary version of this work is presented in \cite{sahin16} as a work-in-progress abstract. This is an extended version of the paper \cite{Sahin16MSC}  to appear in the Proceedings of the IEEE Multi-Conference on Systems and Control (MSC), Buenos Aires, Argentina, September 19-22, 2016. This work is supported in part by NSF grants CNS-1446298 and ECCS-1553873, and DARPA grant N66001-14-1-4045.}
\thanks{The authors are with the Electrical Engineering and Computer Science Department, University of Michigan, Ann Arbor, MI 48109, email: {\tt\small ysahin,necmiye@umich.edu}. }}% <-this % stops a space
\begin{document}

\maketitle
\thispagestyle{empty}
\pagestyle{empty}

\begin{abstract} 
\input{0_abstract}

\end{abstract}


% \fi
%\IEEEpeerreviewmaketitle

\section{Introduction}

\input{1_introduction}

\section{Preliminaries}\label{sec:prelim}
\input{2_prelim}

%\section{Mathematical framework}\label{sec:framework}
\input{3_framework}

\section{Problem Statement}
\input{4_problem_statement}

\section{Algorithm}
\input{5_algorithm}

\section{Illustrative Examples}\label{sec:ex}
\input{6_examples}

\section{Case Study}
\input{7_eps_desc}

\section{Conclusion}
\input{8_conclusion}

%\section*{Appendix}
\appendix
\input{9_appendix}

\bibliographystyle{abbrv}
\bibliography{ref.bib}

\end{document}

%% file: 0_abstract.tex
%!TEX root = dist_vms.tex
This paper considers the synthesis of distributed reactive control protocols for a Boolean network in a distributed manner. We start with a directed acyclic graph representing a network of Boolean subsystems and a global contract, given as an assumption-guarantee pair. Assumption captures the environment behavior, and guarantee is the requirements to be satisfied by the system. Local assumption-guarantee contracts, together with local control protocols ensuring these local contracts, are computed recursively for each subsystem based on the partial order structure induced by the directed acyclic graph. By construction, implementing these local control protocols together guarantees the satisfaction of the global assumption-guarantee contract. Moreover, local control protocol synthesis reduces to quantified satisfiability (QSAT) problems in this setting. We also discuss structural properties of the network that affect the completeness of the proposed algorithm. %The algorithm is quite generic allowing generalizations to continuous or infinite horizon problems. %We also highlight computational bottlenecks in the algorithm that could hamper possible generalizations.
As an application, we show how an aircraft electric power system can be represented as a Boolean network, and we synthesize distributed control protocols from a global assumption-guarantee contract. The assumptions capture possible failures of the system components, and the guarantees capture safety requirements related to power distribution.

%Satisfiability, compositional synthesis, boolean networks, contract-based design

%This work considers the problem of synthesis of reactive control protocols for a Boolean network in a distributed manner. We start with a directed graph representing the network and a global assumption-guarantee contract where assumptions are the constraints on the external variables and the guarantees are the requirements to be satisfied by the system. The distributed control architecture is computed by finding the strongly-connected components of the graph and it is partially ordered by construction. Local contracts for each component are computed such that implementing local controls on them ensures the satisfaction of the global assumption-guarantee pair. The algorithm is correct-by-construction and it is demonstrated on an aircraft electric distribution system.

%% file: 1_introduction.tex
%!TEX root = dist_vms.tex

Cyber-physical systems should be able to react to internal and external events since they are expected to work in dynamic settings. Ability to algorithmically synthesize reactive control protocols that can guarantee correct and safe behavior of such systems, can reduce the cost and efforts spent in extensive simulation and testing phase. This fact motivated the research on correct-by-construction control synthesis approaches that start with accurate system models and formal system requirements and that algorithmically construct the controllers to achieve desired closed loop behavior \cite{tabuada2009verification}. However, similar principled approaches for synthesis of distributed controllers are scarce partly due to the fact that general distributed synthesis problem is computationally hard \cite{pnueli90} when temporal logic requirements are considered.

There is a relation between the control architecture and difficulty of the distributed synthesis problem \cite{FinkbeinerS05}. Recently some fragments of linear temporal logic (LTL) \cite{chatterjee2013distributed} and some control architectures  \cite{madhusudan01} that render the synthesis problem decidable have been identified. However, efficient algorithms for the decidable fragments are somehow lacking \cite{chatterjee2013distributed}. In this work, we consider a simplified version of the distributed control synthesis problem where the system is modeled as a Boolean network and specifications are given as propositional formulae on environment and system variables. This is a decidable problem since the set of all the controllers is enumerable. It can also be seen as a very special case of distributed LTL synthesis with invariance specifications and memoryless system models. We start by showing that the problem is generally hard (i.e., NEXPTIME-complete) even in this simplified setting. We then propose a sound algorithm to synthesize distributed controllers in a relatively efficient way. The algorithm is shown to be complete when the specification and the control architecture satisfy certain structural properties.

%Similar harness results are also common in stochastic decentralized control \cite{witsenhausen1968counterexample}. There is also some structural properties known to make the problem easier in the stochastic setting \cite{witsenhausen1971separation,yoshikawa1978decomposition}

Boolean networks are introduced in \cite{KAUFFMAN1969437} to model and analyze the gene circuits. Recent results on control and analysis of such networks can be found in \cite{cheng2010analysis}. However, to the best of our knowledge, distributed control problem has not been addressed before in this domain.
Our algorithm is motivated by interface theories for stateless components \cite{de2001interface} and contract-based verification techniques \cite{Benvenuti08,bozzano2014formal}. However, instead of verification, we consider controllable components and synthesis of controllers. In particular, the proposed algorithm starts from a global assume-guarantee specification and maps it to local assume-guarantee specifications for individual subsystems. Finding these local assume-guarantee specifications leads to a modular design framework. For instance, local controllers for subsystems can be synthesized independently via SAT or a local subsystem can be replaced with another subsystem that satisfies the same local assume-guarantee specifications. 

In the second part of the paper, we consider an application of the proposed approach to the co-synthesis of a distributed control architecture and correct-by-construction control protocols for an aircraft electric power system (EPS) \cite{nuzzocontract}. We show how the steady state behavior of an electric power system circuit can be modeled as a Boolean network and synthesize distributed controllers using the proposed algorithm.

%% file: 2_prelim.tex
%!TEX root = dist_vms.tex

\subsection{Notation}
The set of positive integers up to and including $n$ is denoted by $\mathbb{N}_n$. For a set $K$, its cardinality is denoted by $|K|$, and its complement is denoted by $K^c$. The Boolean domain $\{True,False\}$ is denoted by $\mathbb{B}$. Given a set $X$ of Boolean variables, the set of all valuations is denoted by $V_x = \mathbb{B}^{|X|}$ and called the domain of $X$, and a specific valuation is denoted by $x\in V_x$.

  \subsection{Graph theory}

A \emph{graph} is a pair $\mathcal{G}=(\Node,\Edge)$, where $\Node$ is a set of nodes (vertices) and $\Edge\subset (\Node\times \Node)$ is a set of edges.  

The graph $\mathcal{G}$ is called \emph{directed} (digraph) when edges have a direction. We say that edge $e = (u,v) \in \Edge$ points from $u$ to $v$. A \emph{path} $\sigma$ on a digraph is a sequence $\sigma= \{e_1,\ldots,e_k\}$ of consecutive edges $e_i = (\node_s^i, \node_d^i)$, i.e., with $\node_d^{i} = \node_s^{i+1}$ for all $1\leq i < k$.  The path $\sigma$ is said to be a path from node $\node_s^1$ to node $\node_d^k$. A path with identical start and end nodes (i.e., $\node_s^1 = \node_d^k$) is called a \emph{cycle}. If directed graph $\mathcal{G}$ has no cycles, it is called a \emph{directed acyclic graph}.
Furthermore a directed acyclic graph is called a \emph{tree} if any two vertices are connected with at most one path. Any disjoint union of trees is called a \emph{forest}. A node without any outgoing edge is called a \emph{leaf node}.

A graph $\mathcal{G}'=(\Node',\Edge')$ is called an \emph{induced subgraph} of $\mathcal{G}$ if $\Node'\subseteq \Node$ and $\Edge' \subset (  (\Node'\times\Node') \cap \Edge)$. All the subgraphs we consider in this paper are induced subgraphs, so, we call them just subgraph for short. A digraph is called \emph{strongly connected} if there is a path from each node to every other node. The \emph{strongly connected components} of a directed graph $\mathcal{G}$ are its maximal strongly connected subgraphs. Here maximal is used in the sense that no strongly connected component of $\mathcal{G}$ is a subgraph of some other strongly connected subgraph of $\mathcal{G}$.
 
%The \emph{condensation} $\mathcal{G}'$ of a digraph $\mathcal{G}$ is the digraph obtained by replacing each strongly connected component of $\mathcal{G}$ with a single node and all edges from one strongly connected component to another one with a single edge having the same direction. %An example of a digraph and its condensation is shown in Fig. \ref{fig:largecircuit} (Left) and (Right). An important fact about the condensations is that they are always directed acyclic graphs.

%% file: 3_framework.tex
%!TEX root = dist_vms.tex

\subsection{Boolean systems and networks}

A \emph{Boolean system} $S$ is a tuple of the form $\left<U, E, Y, f\right>$ where
\begin{itemize}
\item $U = \{u^{(1)},\ldots,u^{(n_u)}\}$ is the set of Boolean control inputs with domain $V_u \doteq \mathbb{B}^{n_u}$,
\item $E= \{e^{(1)},\ldots,e^{(n_e)}\}$ is the set of Boolean environment (uncontrolled) inputs, disjoint from $U$, with domain $V_{e}\doteq \mathbb{B}^{n_e}$,
\item $Y = \{y^{(1)},\ldots,y^{(n_y)}\}$ is the set of Boolean outputs with domain $V_y\doteq \mathbb{B}^{n_y}$, and
\item $f:V_u \times V_e \rightarrow V_y$ is the system function, that maps the inputs to the outputs.
\end{itemize}
The $j^{th}$ component of the system function is denoted by $f^{(j)} : V_u \times V_e \rightarrow V_{y^{(j)}}$. Note that this system definition corresponds to a memoryless or stateless system, that is,
\beq \label{eq:y=fue}
y = f(u,e)
\eeq
where $y\in V_y$, $u\in V_u$ and $e \in V_e$ are valuations of the outputs, control inputs and environment inputs, respectively. 

%\subsection{Boolean Networks}

Given two systems $S_1=\left<U_1, E_1, Y_1, f_1\right>$ and $S_2=\left<U_2, E_2, Y_2, f_2\right>$, a \emph{serial interconnection} from $S_1$ to $S_2$ is formed by equating a set of outputs of $S_1$ to a set of environment inputs of $S_2$. We denote the set of shared variables in a serial interconnection as 
\beq\label{eq:shared}
{I_{1,2}} = \{(k,l)\mid y_1^{(k)} = e_2^{(l)}\},
\eeq
and say $S_1$ is connected to $S_2$ when the set of shared variables is nonempty, i.e., ${I_{1,2}}\neq\emptyset$. 

%For $E_2^l \in E_2$, $E_2^l$ is called an internal input

A \emph{Boolean network} is a tuple $S = \langle \{S_i\}_{i=1}^n, \mathcal I \rangle$ that consists of a collection of subsystems $S_i=\left<U_i, E_i, Y_i, f_i\right>$, for $i\in\mathbb{N}_n$, and an interconnection structure $\mathcal I =\{I_{i,j}\}_{i,j\in\mathbb{N}_n}$, where $I_{i,j}$ is defined as in \eqref{eq:shared}. The interconnection structure represents which subsystems are connected to which others through which variables. The interconnection structure induces a digraph $\mathcal G_S = (\Node_S,\Edge_S)$ called the \emph{system graph}, where $\Node_S = \{S_1,\ldots, S_n\}$ and $\Edge_S=\{(S_i,S_j)\mid {I_{i,j}}\neq\emptyset\}$. The set $E^{int}_i \subseteq E_i$ of internal inputs for a system $S_i$ is defined as $E^{int}_i = \{e_i^{(l)} \in E_i \mid \exists j\in \mathbb{N}_n, k\in \mathbb{N}_{n_{y_j}}, (k,l)\in I_{j,i}\}$. Remaining inputs of  $S_i$ are called external inputs and defined as $E_i^{ext} = E_i\setminus E^{int}_i$. 
%For a Boolean network we denote the set of all control inputs, environment inputs and outputs as $U \doteq \bigcup U_i$, $E \doteq \bigcup E_i$, and $Y \doteq \bigcup Y_i$.

For well-posedness of the network, we assume the following: (i) If two systems are connected to a third one, they are connected through different environment inputs. That is, if there exist some $i,j\in\mathbb{N}_n$, $i\neq j$ with $(k_1,l) \in I_{i,j}$, then there does not exist any $k_2$ with $(k_2,l) \in I_{i',j}$ for any $i'\neq i$. (ii) The system graph $\mathcal G_S$ is a directed acyclic graph. Note that in the memoryless Boolean setting, cycles in the interconnection structure lead to mostly not well-defined algebraic loops, therefore such interconnections are not considered. In fact, the Boolean network $S$ itself is a Boolean system whose inputs, outputs and system function can be derived from those of its subsystems, and the interconnection structure.

\subsection{Specifications}

We consider specifications in the form of \emph{assumption-guarantee} pairs given in terms of Boolean-valued functions. This section introduces some terminology regarding Boolean-valued functions and assume-guarantee specifications.

Let  $\varphi: V_x \rightarrow \mathbb{B}$ be a Boolean-valued function of variables in $X$. Each $\varphi$ can be equivalently represented by its satisfying set $\llbracket \varphi \rrbracket \subseteq V_x$ defined as follows:
\beq
\llbracket \varphi\rrbracket \doteq \{x \in V_{x} \mid \varphi(x)=True  \}.
\eeq
Propositional operations on Boolean-valued functions are defined in the usual way. The symbols $\neg,\land$ and $\lor$ are used for logical operations \emph{NOT} (negation), \emph{AND} (conjunction) and \emph{OR} (disjunction), respectively and performs as follows:
$\llbracket \neg \psi \rrbracket \doteq \llbracket \psi \rrbracket^c$, $\llbracket \psi_1 \wedge \psi_2 \rrbracket \doteq \ws{\psi_1} \cap \ws{\psi_2}$, and $\llbracket \psi_1 \lor \psi_2 \rrbracket \doteq \ws{\psi_1} \cup \ws{\psi_2}$.

Additional operations such as \emph{XOR} (exclusive-or) denoted by $\oplus$, can be defined using the main operators above: $ \psi_1 \oplus \psi_2 \doteq (\psi_1 \land \neg\psi_2) \lor (\neg\psi_1 \land \psi_2)$.

\begin{defn}	
	Let $\varphi: V \rightarrow \mathbb{B}$ be a Boolean-valued function and assume $V = \prod_{i \in J} V_i$. % and $(x_1, \dots 
	The \emph{projection} of $\varphi$ onto a set $I\subseteq J$ of variables is another Boolean function $\varphi |_{I} : \prod_{i \in I} V_{i} \rightarrow \mathbb{B}$ whose satisfying set is defined as  
	$$\llbracket \varphi |_{I}\rrbracket \doteq \left\{ x \in \prod_{i \in I} V_{i} \mid \exists y \in \prod_{j \notin I} V_{j} \text{ such that } (x,y)\in \llbracket \varphi \rrbracket \right\}.$$ 
\end{defn}

When the identity of the variables are clear from the context, the order of the tuple is ignored. In other words, saying $(x,y)\in \llbracket \varphi \rrbracket$ is equivalent to saying $(y,x)\in \llbracket \varphi \rrbracket$ and vice versa. 

Assumptions and guarantees are Boolean functions that capture the  \emph{a priori} knowledge about the uncontrolled inputs and the desired safe behavior of the system outputs, respectively.

%{\color{blue}
\begin{defn} An \emph{assumption} $\Asmp: V_{e} \rightarrow \mathbb{B}$ for a Boolean system $S$ is a Boolean-valued function of its uncontrolled inputs. \end{defn}
We say that $(e_1^{ext},\dots,e_i^{ext},\dots,e_{n}^{ext}) \in V_e$ is \emph{admissible} if it is in the satisfying set of $\Asmp$. With a slight abuse of terminology, we also say $e_i^{ext}$ is admissible, for the sake of convenience. % if there exists a tuple $(e_1^{ext},\dots,e_i^{ext},\dots,e_{n}^{ext}) \in \llbracket \Asmp \rrbracket$.

\begin{defn} A \emph{guarantee} $\Gar: V_{y} \rightarrow \mathbb{B}$ for a Boolean system $S$ is a Boolean-valued function of its outputs. %If $S$ is a Boolean network with subsystems $S_1,\ldots, S_{n}$, then $V_y = \prod_{i=1}^{n} V_{y_i}$.
\end{defn}

For a Boolean network $S$ with subsystems $S_1,\ldots, S_{n}$, by convention, $\Asmp$ denotes an assumption with domain $V_{e} = \prod_{j=1}^{n} V_{e_j^{ext}}$. Furthermore,  
$\Asmp^{\downarrow(i)}$ %denotes an assumption associated with subsystem $S_i$ with domain, possibly, $V_{e} = \prod_{j=1}^{n} V_{e_j^{ext}}$, and $\Asmp^{\downarrow (i)}$ 
denotes an assumption associated with subsystem $S_i$ with domain $V_{e_i^{ext}}$. Note that for a given assumption $\Asmp$, its projection $\Asmp|_{E^{ext}_i}$ can always be denoted by $\Asmp^{\downarrow(i)}$ as it only contains variables from $E_i$. Similarly, $\Gar$ denotes a guarantee with domain $V_{y} = \prod_{j=1}^{n} V_{y_j}$ and $\Gar^{\downarrow(i)}$ denotes a guarantee associated with subsystem $S_i$ with domain $V_{y_i}$. 

\begin{defn}
	A formula of the form
	\beq\label{eq:formula_global}
	\varphi\doteq \bigwedge_{k=1}^{n_c}(\asmp{k}{}\to\gar{k}{})
	\eeq
	is called a \emph{global contract} and denoted by $\mathbb{C} \doteq \{[A_k,G_k]\}_{k=1}^{n_c}$. 
\end{defn}

 \begin{defn}
 	A formula of the form
	\beq
	\varphi^{(i)}\doteq \bigwedge_{k=1}^{n_c}(\asmp{k}{\downarrow(i)}\to\gar{k}{\downarrow(i)})
	\eeq
 	is called a \emph{local contract} for $S_i$ and denoted by $\mathbb{C}_i \doteq \{[\asmp{k}{(i)},\gar{k}{\downarrow(i)}]\}_{k=1}^{n_c}$. 
 \end{defn}
 
\subsection{Control protocol synthesis}
Given a system $S$, a \emph{control protocol} $\pi: V_{e} \rightarrow V_u$ maps the environment inputs to control inputs. When a control protocol $\pi$ is implemented on a system $S$, \emph{controlled system} is governed by the following input-output relation:
\beq \label{eq:sys_sr}
y = f(\pi(e),e),
\eeq
since $u= \pi(e)$.

Given an assumption-guarantee pair, the control synthesis problem aims to find a protocol $\pi$ that sets the control inputs such that $A \rightarrow G$ is satisfied. For any $e \in \llbracket A \rrbracket$, determining if there exist a $u \in V_u$ so that $\Gar$ evaluates $True$, is a quantified Boolean SAT problem: 
\begin{equation}\label{eq:qsat}
\forall e \in V_e: \exists u \in V_u : \Asmp(e) \to \Gar(y)
\end{equation}
where $y$ is given as in the Equation \eqref{eq:y=fue}.
If the quantified SAT problem above can be solved, then the set of all solutions can be used as the control protocol. In this case, we say $A \rightarrow G$ is \emph{realizable}.

%% file: 4_problem_statement.tex
%!TEX root = dist_vms.tex

We are interested in synthesizing distributed controllers for a Boolean network to achieve a common goal while each local controller has access only to its own inputs. The general form of the problem can be formally stated as follows:

\begin{Problem}\label{prob:1}
Given a Boolean network $S = \langle \{S_i\}_{i=1}^n, \mathcal I \rangle$ and a Boolean formula $\varphi$ over the input and output variables $\bigcup_{i=1}^n(E_i\cup Y_i)$, find local controllers $\pi_i: V_{e_i} \to V_{u_i}$ such that when these local controllers are implemented together, the controlled system satisfies $\varphi$ for all possible inputs.  
\end{Problem}

Next, we show that even in this relatively simple Boolean setting, the distributed synthesis problem is quite hard. The main difficulty arises from the fact that each local controller only has partial information to base their decisions on. In order to show the hardness of the problem, we reduce dependency quantifier boolean formula game (DQBFG) to Problem \ref{prob:1}. Let us first introduce DQBFG \cite{hearn2009games}.

%\begin{defn} 
DQBFG is a three player game, consisting of players $B$ (black), $W1$ (white 1) and $W2$ (white 2). An instance of DQBFG is a Boolean formula $\varphi$ in variables $X_1\cup X_2 \cup Z_1 \cup Z_2$. For each $i\in \{1,2\}$, player $Wi$ has access only to the variables in $X_i\cup Z_i$. First, player $B$ chooses an assignment for variables in $X_1\cup X_2$. Then, player $W1$ chooses an assignment for variables in $Z_1$ and player $W2$ chooses an assignment for variables in $Z_2$. Note that the order of the decisions made by $W1$ and $W2$ does not matter since they do not see each other's decisions. White team wins if $\varphi$ is true in the end. The decision problem DQBFG is whether a winning strategy for the white team in this game exists. DQBFG belongs to the complexity class NEXPTIME-complete \cite{hearn2009games}.
%\end{defn}

%The computational complexity of DQBFG NEXPTIME-complete \cite{hearn2009games}.

\begin{Theorem}\label{theo:hardness} The decision version of Problem~\ref{prob:1} is NEXPTIME-complete.
\end{Theorem}
\begin{proof}
	We prove our claim by showing that DQBFG is a special case of our problem. %Let $G_S$ has three nodes, $B,W1,W2$ such that $\Edge =\{(B,W1),(B,W2)\} $. 
	We can imagine environment taking over the role of $B$ in our setting. Boolean network $S$ consists of two subsystems $S_1,S_2$ associated with $W1$ and $W2$ respectively, where 
	\begin{itemize}
		\item $\mathcal{I} \doteq \{I_{1,1},I_{1,2}, I_{2,1}, I_{2,2}\}$ with $I_{i,j} \doteq \emptyset$ for all $i,j = 1,2$, % Let 
		\item $E_i \doteq X_i$ and $U_i \doteq Z_i$ for $i=1,2$,%Finally, assume 
		\item $Y_i \doteq E_i\cup U_i = X_i\cup Z_i$ and the system function is the identity map for each subsystem, i.e., $f_i: (u_i,e_i) \mapsto (u_i,e_i)$ 
		for $i=1,2$.
	\end{itemize}
	The Boolean formula $\varphi'$ is obtained by replacing each variable $X_i$ with $Z_i$ in $\varphi$ with the corresponding variables in $Y_i$.
	Then, the decision version of Problem ~\ref{prob:1}, that is, verifying the existence of local controllers to satisfy a Boolean function $\varphi'$ for the network $S$ as defined above, is equivalent to DQBFG. This proves NEXPTIME-hardness.
	
	Now we show that Problem~\ref{prob:1} is in NEXPTIME. A naive algorithm for Problem~\ref{prob:1} that enumerates all possible local controller tuples $\pi^i$ for $i=1,\ldots, n$ and checks the satisfaction of
	$\varphi$ runs in NEXPTIME since there are $\mathcal{O}(2^{\sum_{i\in{\mathbb N}_n}|U_i|2^{ |E_i|}})$ local controller tuples. 	
	Thus the complexity of the decision version of Problem~\ref{prob:1} is NEXPTIME-complete.
\end{proof}

Given the hardness result above, we seek to break the problem into more manageable pieces. We assume that $\varphi$ is given in the form of a global contract as in \eqref{eq:formula_global}. Note that this is without loss of generality as any $\varphi$ can be written in this form. We consider the following variant of the problem. 

\begin{Problem}\label{prob:2}
Given a Boolean network $S = \langle \{S_i\}_{i=1}^n, \mathcal I \rangle$, and a global contract $\mathbb{C} = \{[\asmp{k}{},\gar{k}{}]\}_{k=1}^{n_c}$, for $i=1,\ldots,n$, find local contracts of the form $\mathbb{C}_i= \{[\asmp{k}{\downarrow(i)},\gar{k}{\downarrow(i)}]\}_{k=1}^{n_c(i)}$ and local controllers  $\pi_i: V_{e_i} \to V_{u_i}$, satisfying these local contracts,  such that when these local controllers are implemented together, the controlled system satisfies $\mathbb{C}$ for all possible inputs.   %accordingly. 
\end{Problem}

Note that Problem \ref{prob:2} is essentially equivalent to Problem \ref{prob:1}. However, if the structure of the problem allows the local contracts $\mathbb{C}_i$ to be computed efficiently, then obtaining a local controller from a local contract is a full information synthesis problem that can be solved via QSAT. Complexity of QSAT is PSPACE-complete, which is lower in the complexity hierarchy than NEXPTIME-complete \cite{papadimitriou2003computational} . Moreover, there are off-the-shelf highly optimized tools for solving QSAT problems making their solution practically feasible. In what follows, we provide a sound algorithm to solve Problem \ref{prob:2}. Some structural properties of $\mathbb{C}$ and the interconnection structure of $S$ that render the algorithm complete are also discussed.

%% file: 5_algorithm.tex
%!TEX root = dist_vms.tex

In this section, we impose a certain structure on 
the specification $\varphi$ so that Problem~\ref{prob:2} can be reduced to a  
number of QSAT problems. In particular, we let the specification $\varphi$ be given by a single assumption-guarantee pair, $\mathbb{C} = [\Asmp,\Gar]$ and propose Algorithm~\ref{alg:main} to solve for this case. 

An overview of the algorithm is as follows.  Let $S_i$ be a subsystem that is a leaf node of the system graph. We first compute a local assumption $\Asmp^{\downarrow(i)}$ (line 3) and a guarantee $\Gar^{\downarrow(i)}$ (line 4). Then, we synthesize the local controller, which is a full information synthesis problem that can be solved by QSAT. If the local assumption-guarantee pair is unsatisfiable, we constrain the internal inputs of $S_i$ to make it realizable using the so called \emph{least restrictive assumptions} (line 6). Since internal inputs of $S_i$ are outputs of other subsystems, this assumption becomes a guarantee the rest of the systems must fulfill (line 8). Once we synthesize the controller, we discard $S_i$ from the system graph and we pick another leaf node $S_j$ and continue in the same manner until every subsystem has a local controller. Different subroutines of the algorithm are explained in detail next.

\begin{algorithm} %\label{alg:main}
	\caption{Distributed\_Synthesis($S,\mathbb{C}$)} \label{alg:main}
	\begin{algorithmic}[1]
		\small
			\Require{A Boolean network $S$ and global contract $\mathbb{C} = [A , G]$} 
					\If{$S$ is empty}
		\Return $True$
		\EndIf
		\State Let $S_{j}$ be a leaf node of $\mathcal{G}_S$
		%\State{Let $S_{leaves}$ be the s}
		%\State $\Psi_g$ = distribute\_guarantees($G_S,\Gar$)
		%\For{$S_i \in S_{leaves}$}
		\State $ \Asmp^{\downarrow(j)} =A|_{E^{ext}_j}$ 
		%\State $\Gar^{\downarrow(j)}= G|_{Y_j} $
		\State $\Gamma$ = distribute\_guarantee($\Gar,S_j$)
		\For{$\gamma \in \Gamma$}
		\State $(\lra^{(j)}_{lra}, \pi_j)$ = find\_lra($S_{j},A^{\downarrow(j)}, G^{\downarrow(j)}$)
		\If{$\lra^{\downarrow(j)}_{lra}$ is $False$}
		\Return $False$
		\EndIf
		\State $\mathbb{C}' = \left[A , \left( G^{\uparrow(j)}\wedge \lra^{(j)}_{lra}\right)\right]$
		%\State $\varphi' = \left( \wedge_{i \neq j} A^{\downarrow(j)}\right) \to \left( \left( \wedge_{i \neq j} G^{\downarrow(j)}\right)\wedge G^{(j)}_{lra}\right)$
		\State $S'$ = delete($S,S_j$)
	%	\EndFor
	
		%\For{$S_{t,0} \in S_{roots}$}
	%	\State $(G^{(t,0)}_{lra}, \pi_i)$ = %find\_lra($S_{t,0}, A^{\downarrow(t,0)},(\wedge_i G^{(t,i)}_{lra} \wedge  G^{\downarrow(t,0)} )$)
	%	\EndFor
			%\Return $ \{\pi_{t,i}\}_{(t,i)} $
\State flag = Distributed\_Synthesis($S',\mathbb{C}'$)
\If{flag}
\Return $True$
\EndIf
\EndFor\\
\Return $False$

	\end{algorithmic}
\end{algorithm}

\subsection{Finding local assumptions}
Local assumptions are simply found by projection as shown in line 3 the Algorithm \ref{alg:main}. By definition of the projection operator, we have $\Asmp^{\downarrow(i)}$ depend only on local external variables. Moreover, projection ensures that local assumptions do not restrict the environment more than the global assumption and does so in the ``best" possible way. In particular, we have the following property that will be useful in proving the soundness of the algorithm.
\begin{prop}\label{theo:as}
	Let $A : \prod_{i=1}^{n} V_{e^{ext}_i} \rightarrow \mathbb{B}$ be the global assumption. 
	Define $S_i$'s local assumption $\Asmp^{\downarrow(i)} \doteq \Asmp|_{V_{e_i^{ext}}}$ for $i \in \mathbb{N}_n$. Then local assumptions are less restrictive than the global assumption, i.e.,
	\begin{equation}\label{eq:asmp}
	\llbracket A \rrbracket \subset \prod_{i=1}^n \llbracket A^{\downarrow(i)}\rrbracket.
	\end{equation}
	Moreover, \eqref{eq:asmp} fails to hold if we replace any ${A}^{\downarrow(i)}$ with $\bar{A}^{\downarrow(i)}$ where $\llbracket \bar{A}^{\downarrow(i)}\rrbracket \subset \llbracket A^{\downarrow(i)}\rrbracket$.
\end{prop}

\begin{proof}
	Assume that 
	$(e_1^{ext},\dots,e_n^{ext}) \in \llbracket A \rrbracket $.
	By construction, 
	$e_i^{ext} \in \llbracket A^{\downarrow(i)} \rrbracket \text { for all } i \in \mathbb{N}_n$.
	Then 
	$
	(e_1^{ext},\dots,\dots,e_n^{ext}) \in  \prod_{i=1}^n \llbracket A^{\downarrow(i)} \rrbracket.
	$
	This implies \eqref{eq:asmp}.
	
	For any arbitrary $i \in \mathbb{N}_n$, let $\bar{A}^{\downarrow(i)}$ be a Boolean-valued function satisfying $\llbracket \bar{A}^{\downarrow(i)}\rrbracket \subset \llbracket A^{\downarrow(i)}\rrbracket$ and $\llbracket \bar{A}^{\downarrow(j)}\rrbracket = \llbracket A^{\downarrow(j)}\rrbracket$ for every other $j\neq i$. Now let $e^{ext}_i \in \llbracket A^{\downarrow(i)}\rrbracket \setminus \llbracket \bar{A}^{\downarrow(i)}\rrbracket$. By construction of $A^{\downarrow(i)}$, there exists an environment valuation $(e_1^{ext},\dots,e^{ext}_i,\dots,e_n^{ext}) \in \llbracket A \rrbracket $. However, $(e_1^{ext},\dots,e^{ext}_i,\dots,e_n^{ext}) \in \prod_{j=1}^n \llbracket \bar{A}^{\downarrow(j)} \rrbracket $. Thus \eqref{eq:asmp} is no longer true.
\end{proof}

\subsection{Finding local guarantees}
As opposed to the local assumptions, local guarantees $\Gar^{\downarrow(i)}$ cannot be less restrictive than the global guarantee. No communication is assumed between subsystems, hence any possible combination of outputs allowed by local guarantees should be in the satisfying set of global guarantee, i.e.,
\beq\label{eq:dist}
\forall i \in \mathbb{N}_n: \forall y_i \in \ws{\Gar^{\downarrow(i)}}:  (y_1,\dots,y_n)  \in \ws{G}.
\eeq
 % that once the controllers ensuring the local guarantees are put together, we can ensure that the global guarantee is satisfied. 
 
Due to the dependence, local guarantees cannot be computed independently. Our algorithm proceeds by selecting a local guarantee for a subsystem at each iteration, using the notion of a distribution, which is introduced next.

\begin{defn}
Let $G : \prod_{j=1}^n V_{y_i} \rightarrow \mathbb{B} $ be the global guarantee. A \emph{distribution} of $G$ between $S_i$ and the rest of the subsystems is a pair of Boolean valued functions $ \Gar^{\downarrow(i)}: V_{y_i} \to \mathbb{B}$ and $ \Gar^{\uparrow(i)}:\prod_{j \neq i} V_{y_j} \to \mathbb{B}$ where
\begin{equation}\label{eq:dist2}
\forall y_i \in \ws{\Gar^{\downarrow(i)}}: \forall y \in \ws{\Gar^{\uparrow(i)}}: (y_i,y)  \in \ws{G}.
\end{equation}
We denote a distribution as $\gamma = \{\Gar^{\downarrow(i)},\Gar^{\uparrow(i)}\}$.
%For a distribution, the set $\gamma = \{\Gar^{\downarrow(i)},\Gar^{\uparrow(i)}\}$ is called the \emph{local guarantee set}.
\end{defn}

Unfortunately, in general, there is no unique ``best" way of generating local guarantees, therefore distributions are not unique. Note that 
$
 \llbracket\Gar^{\downarrow(i)}\rrbracket \times \ws{\Gar^{\uparrow(i)}}\subset \llbracket G\rrbracket
 $ 
is trivial by Eq.~\eqref{eq:dist2}. This implies that local guarantees are conservative and they under-approximate $\ws{G}$.
However, we do not want to restrict the system more than necessary. That is why we compute only the \emph{maximal distributions}. A distribution is called \emph{maximal} when there is no other distribution $\bar{\gamma} = \{\bar{\Gar}^{\downarrow(i)},\bar{\Gar}^{\uparrow(i)}\}$ that satisfy $\ws{\Gar^{\downarrow(i)}}\subset \ws{\Bar{\Gar}^{\downarrow(i)}}$ and/or $\ws{\Gar^{\uparrow(i)}}\subset \ws{\Bar{\Gar}^{\uparrow(i)}}$.
%This implies that maximal distributions give the least conservative under-approximation of $\ws{G}$. 
We denote the set of all maximal distributions with $\Gamma = \{\gamma_k\}_k$, which is computed in line 4 of the algorithm. We refer the reader to the appendix for more details on how to compute these distributions.

\subsection{Least restrictive assumptions and controller synthesis}\label{sec:32}

 Controller synthesis for any subsystem %$S_i$ for a given assumption-guarantee pair $[\Asmp^{\downarrow(i)} \to \Gar^{\downarrow(i)}]$ 
 is essentially a quantified SAT problem as stated in \eqref{eq:qsat}. %If for all 
% \begin{equation}\label{eq:lra}
% \forall e_i^{ext} \in V_{e_i^{ext}}: \exists u_i \in V_{u_i} : \Asmp^{\downarrow(i)} \to \Gar^{\downarrow(i)}.
% \end{equation}
However, if we let internal inputs to take any possible value, then it might not be possible to find an input $u_i$ that renders $G^{\downarrow(i)} = True$. %Fortunately, internal inputs can be controlled by $S_i$'s ancestors. 
When this is the case, we restrict the internal inputs, which are controlled by $S_i$'s ancestors, to a certain set to achieve $G^{\downarrow(i)}$. While doing so, we would like to be as permissive as possible.

 \begin{defn}
 	Given a local contract $\mathbb{C}_i = [\Asmp^{\downarrow(i)} , \Gar^{\downarrow(i)}]$, the set of all internal inputs that makes the contract realizable is called the \emph{least restrictive assumption}. It is denoted with $\lra^{(i)}_{lra}$ where 
 	\beq
 	\llbracket \lra^{(i)}_{lra} \rrbracket \doteq \{ e_i^{int} \in E_i^{int} \mid  \Asmp^{\downarrow(i)} \to \Gar^{\downarrow(i)} \text{ is realizable}\}.
 	\eeq 
\end{defn}
\vspace{2mm}
 
 In other words, the least restrictive assumption gives the set of internal inputs that makes the guarantee realizable. Any internal input outside of this set makes the guarantee unsatisfiable. 

After computing the least restrictive assumption, we update the local contract as
\beq
\mathbb{C}_i = \left[ \left(A^{\downarrow(i)} \wedge \lra^{\downarrow(i)}_{lra}\right),G^{\downarrow(i)} \right].
\eeq

Note that by definition of least restrictive assumption, $\mathbb{C}_i$ is realizable. Having a realizable local contract, the control protocol is synthesized by solving the respective QSAT problem. Line 6 in Algorithm \ref{alg:main} performs these two operations simultaneously.

On the other hand, the least restrictive assumption $\lra^{(i)}_{lra}$ imposes new guarantees to the ancestors of $S_i$. We change internal inputs with their output correspondents by examining the interconnection structure $\mathcal{I}$. Finally we update the global contract for the remaining subsystems as
 
\beq
\mathbb{C}' = [A ,( G^{\uparrow(i)} \wedge \lra^{(i)}_{lra})].
\eeq
 
 \subsection{Algorithm Analysis}
 
 In this section we show that solutions returned by Alg.~\ref{alg:main} are correct. Then we introduce conditions on the system graph and the specifications that renders the algorithm complete. Finally we discuss the complexity of the proposed method. %We show that compared to the original problem complex %We start with the soundness property.
 
 \begin{Theorem}[Soundness]
 	For arbitrary Boolean network $S$, if the specification is given with $\mathbb{C} =[\Asmp , \Gar]$, then Alg.~\ref{alg:main} is sound.
 \end{Theorem}
 \begin{proof}
 	It is enough to show that local contracts results in weaker assumptions on environment inputs and stronger restrictions on system outputs.
		It can be shown using Boolean algebra that 
		\beq\label{eq:sound0}
		\left(\bigwedge_{i=1}^n \left(\Asmp^{\downarrow(i)} \to \Gar^{\downarrow(i)}\right)\right) \to \left(\bigwedge_{i=1}^n \Asmp^{\downarrow(i)} \to \bigwedge_{i=1}^n \Gar^{\downarrow(i)}\right)
		\eeq 
		is a tautology.
		Also from Eq.~\eqref{eq:asmp} and Eq.\eqref{eq:dist}, we can show that $A \to \bigwedge_{i=1}^n \Asmp^{\downarrow(i)}$ and $\bigwedge_{i=1}^n \Gar^{\downarrow(i)} \to G$ are tautologies. Therefore,
		\beq
				\left(\bigwedge_{i=1}^n \Asmp^{\downarrow(i)} \to \bigwedge_{i=1}^n \Gar^{\downarrow(i)} \right) \to \left(A\to G\right)
 		\eeq
is trivially true.

%		 implies that 
%		\beq\label{eq:sound1}
%		A \to \bigwedge_{i=1}^n \Asmp^{\downarrow(i)}
%		\eeq
%		and
%		 %implies that local contracts are less restrictive.
% 	 Eq.\eqref{eq:dist} implies that
% 		\beq\label{eq:sound2}
% 		\bigwedge_{i=1}^n \Gar^{\downarrow(i)} \to G
% 		\eeq
% 		are tautologies. 
 		%Using \eqref{eq:sound0},\eqref{eq:sound1} and \eqref{eq:sound2}, we can see that
 		
 Also note that the least restrictive assumptions computed for each subsystem are not restrictions on environment. In fact they restrict the outputs, which only strengthens result presented above. Hence, if the algorithm returns a controller that satisfies the local contracts, the global contract is satisfied.
 		\end{proof}
 
% \vspace{2mm}
 The completeness requires additional assumptions on the system graph and the specification.
 \begin{Theorem}[Completeness]\label{theo:comp} 
  Let the specifications given with the contract $\mathbb{C}=[\Asmp,\Gar]$ and the system graph induced by the Boolean system $S$ be $\mathcal{G}_S$.
  Alg.~\ref{alg:main} is complete if
  \begin{enumerate}
  	\item $\Asmp = \bigwedge_i A^{\downarrow(i)}$,%
  	\item $\Gar = \bigwedge_i G^{\downarrow(i)}$, and	\item $\mathcal{G}_S$ is a forest.
  	\end{enumerate} 
   \end{Theorem}
 In other words the specification is given as
\beq \label{eq:conj}
\mathbb{C} = \left[\bigwedge_i A^{\downarrow(i)} , \bigwedge_i G^{\downarrow(i)}\right].
\eeq
%\vspace{1mm}

 \begin{proof}
We prove by induction that if Alg.~\ref{alg:main} fails to return a protocol then there does not exist one.
 
When the number of subsystems is one, centralized and distributed algorithms are identical. If the respective QSAT problem is not satisfiable, then there does not exist any control protocol to achieve the task. 

Assume that the Alg.~\ref{alg:main} is complete for an arbitrary forest $\bar{\mathcal{G}}_S$ with $n$ nodes and any specification $\bar{\varphi}$ given in the form of Eq.~\eqref{eq:conj}.

Also let ${\mathcal{G}}_S$ be an arbitrary forest with $(n+1)$ nodes and let 
\begin{equation}
{\mathbb{C}} =  \left[ \bigwedge_{i=1}^{n+1} {A}^{\downarrow(i)} ,  \bigwedge_{i=1}^{n+1} {G}^{\downarrow(i)}\right]
\end{equation}
be its specification.
%where $\bar{A}^{\downarrow(i)} =A^{\downarrow(i)}$ and ${A}^{\downarrow(i)} =A^{\downarrow(i)}$ for $i\in \mathbb{N}_n$.
Without loss of generality, assume that $S_{(n+1)}$ is a leaf node and the local contract that is computed according to Alg.~\ref{alg:main} is given as
$$\mathbb{C}_{(n+1)} = \left[\left({A}|_{V_{E_{n+1}^{ext}}} \wedge {\lra}^{(n+1)}_{lra}\right), {G}^{\downarrow(n+1)} \right].$$

First assume that ${\lra}^{(n+1)}$ is $False$ and the local contract $\mathbb{C}_{(n+1)}$ is not satisfiable. This means that there exists at least one admissible environment valuation $e_{(n+1)}^{ext}$ such that no matter what the internal inputs are, ${G}^{\downarrow(n+1)}$ is unsatisfiable. Since the environment inputs are uncontrolled, no control protocol can overcome this problem. Thus no distributed (or central) control protocol exists for the given system and specifications. 

Now assume that ${\lra}^{(n+1)}$ is not $False$ and the local contract $\mathbb{C}_{(n+1)}$ is satisfiable. Then the global contract is updated as 
%	\begin{equation}\label{eq:update}
%		{\mathbb{C}}' =  \left[\left( \wedge_{i=1}^{n+1} {A}^{\downarrow(i)}   \right) , \left( \left( \wedge_{i=1}^{n+1} {G}^{\downarrow(i)}\right) \wedge {G}^{\downarrow(n+1)}_{lra}\right)\right]
%	\end{equation}

\begin{equation}\label{eq:update}
{\mathbb{C}}' =  \left[{A} , \left({G}^{\uparrow(n+1)}_{lra}\wedge {\lra}^{\downarrow(n+1)}_{lra}\right)\right].
\end{equation}

Note that distribution is unique and ${G}^{\uparrow(n+1)}_{lra} = \bigwedge_{i \neq n+1} {G}^{\downarrow(i)}$.
Let $S_j$ denote the parent of $S_i$. Then ${\lra}^{(n+1)}_{lra}$ is an additional guarantee that involves variables only from $S_j$. %This means that the form of the guarantees are preserved. Then there exists a Boolean formula 
Now define $({G}')^{\downarrow(i)} = {G}^{\downarrow(i)}$ for $i \neq j$ and  $({G}')^{\downarrow(j)} = {G}^{\downarrow(j)} \wedge {\lra}^{(n+1)}_{lra}$. Then we can write
%Now set ${G}^{\downarrow(j)} = {G}^{\downarrow(j)} \wedge  {G}^{\downarrow(n+1)}_{lra}$. Then 
%we can write (\ref{eq:update}) as 
%	\begin{equation}
%	{\varphi}' = \left( \wedge_{i=1}^{n+1} {A}^{\downarrow(i)}   \right) \to \left( \wedge_{i=1}^{n+1} {G}^{\downarrow(i)}\right)
%	\end{equation}

\begin{equation}\label{eq:comp2}
\mathbb{C}' = \left[\bigwedge_{i=1}^{n} ({A}')^{\downarrow(i)} , \bigwedge_{i=1}^{n} ({G}')^{\downarrow(i)}\right]
\end{equation}
where $A' = A|_{\prod_{j \neq i}^{n}}$.
Now we are left with an arbitrary forest with $n$ nodes and a specification given in the form of Eq.~\eqref{eq:conj}. Thus the Alg.~\ref{alg:main} is complete.
 \end{proof}
 \vspace{2mm}
 
% \begin{rem} A few remarks on the complexity of the algorithm are in order. The complexity of the algorithm in general is $\mathcal{O}(2^{2n})$, where $n$ is the number of input, output and environment variables. Under the assumptions of Theorem \ref{theo:comp}, it reduces to $\mathcal{O}(2^{n})$ since the distribution can simply be computed using projection and is unique. Recall that the complexity of the general problem was $\mathcal{O}(2^{(2^n)})$; therefore with the assumed structure, one of the exponents is eliminated.
% \end{rem}
 
% \begin{rem}
% 	A few remarks on the complexity of the algorithm are in order. Proposed algorithm requires solving $\mathcal{O}(n 2^{m+k})$ SAT problems in the worst case, where $n$ is the number of subsystems,  $m = \Sigma_i|Y_i|$ is the number of all outputs and $k = \max_i{|E_i|}$ with $E_i$ is the number of  inputs of $i^{th}$ subsystem. Under the assumptions of Theorem \ref{theo:comp}, it  reduces to $\mathcal{O}(n 2^{k})$ since the distribution can simply be computed using projection and is unique. Recall that the complexity of SAT is $2^n$ where the general problem was $\mathcal{O}(2^{(2^n)})$; therefore with the assumed structure, one of the exponents is eliminated.
% \end{rem}
 
\begin{rem}
	A few remarks on the complexity of the algorithm are in order. In the worst case, the proposed algorithm requires solving $\mathcal{O}(n 2^{\Sigma_i|Y_i|+\max_i{|E_i|}})$ SAT problems each with complexity $\mathcal{O}(2^{\max_i |U_i|})$, where $n$ is the number of subsystems. Under the assumptions of Theorem \ref{theo:comp}, it  reduces to $\mathcal{O}(n 2^{\max_i{|E_i|}})$ since the distribution can simply be computed using projection and is unique. Recall that the complexity of the general problem is $\mathcal{O}(2^{\sum_{i\in{\mathbb N}_n}|U_i|2^{ |E_i|}})$; therefore with the assumed structure, one of the exponents is eliminated.
\end{rem}
 
%In order for algorithm to be complete we need three conditions. When the first conditon in Theorem \ref{theo:comp} is satisfied, \eqref{eq:asmp} becomes an equality:
%$$\llbracket A \rrbracket = \prod_{i=1}^n \llbracket A^{\downarrow(i)}\rrbracket.$$
%
%In other words, Boolean projection does not result in information loss. When this condition is violated, a subsystem can use the its environment inputs to infer another subsystem's. This information then can be used to signal for signaling. 
%Example \ref{ex:asmp} in section \ref{sec:ex} illustrates such a case. 
%
%Similarly, when the second condition is violated, \eqref{eq:gar} becomes a strict inequality. This implies that Boolean distribution is conservative. Example \ref{ex:gar} in section \ref{sec:ex} shows how this results in incomplete solutions.
%
%Finally, if a subsystem has multiple parents, the least restrictive assumption might lead to violation of the second condition. This is demonstrated by example \ref{ex:forest} in section \ref{sec:ex}.

%% file: 6_examples.tex
%!TEX root = dist_vms.tex

\begin{figure}[ht]
  \centering
      \subfigure{\includegraphics [width=0.35\textwidth] {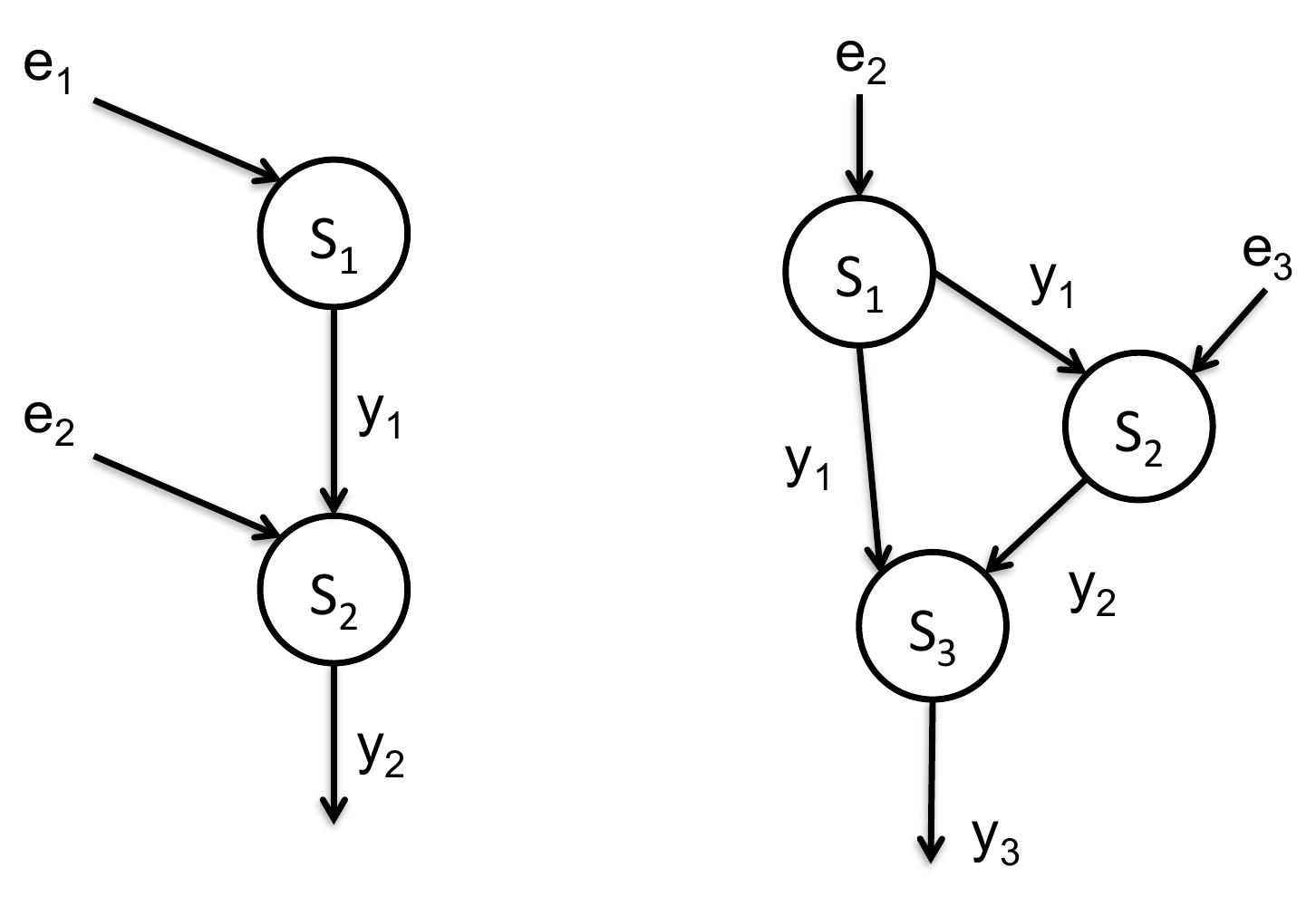}} 
  \caption{(Left): A system graph satisfying the forest condition. (Right): System graph where forest condition is violated.}\label{fig:ex1}
\end{figure}

In this section we present examples to illustrate the proposed method. We use Algorithm~\ref{alg:main} to solve Example \ref{ex:signaling} which satisfies all conditions for completeness.
  
\begin{example}\label{ex:signaling}
		For the system graph seen in Fig.~\ref{fig:ex1} (Left), let the global contract be $\mathbb{C} = [e_1 , y_2]$, where
\begin{align*}
f_1 &: (e_1,u_1) \mapsto u_1\\
f_2&: (e_2,y_1,u_2) \mapsto ((e_2\lor y_1) \wedge u_2).
\end{align*}
\end{example}

In order to synthesize local controllers, we start by finding the local assumption for $S_2$. Using Boolean projection, it can be found as	
$$ \Asmp^{\downarrow(2)} = \Asmp|_{V_{e^{ext}_2}} = True.$$

Note that $e_2$ can take any possible value because the global contract does not put any restrictions on $S_2$'s environment inputs. This implies that the local assumption is trivially true.
Then we distribute the guarantee $\Gar = y_2$ and compute the distribution. In this example guarantee depends only on $S_2$'s inputs, thus 
the solution is unique and it can be written as
$$\Gar^{\downarrow(2)} = y_2 ,\quad \Gar^{\uparrow(2)} = True. $$
Having local assumption-guarantee pair, we compute the least restrictive assumption as $\lra_{lra}^{(2)} = y_1.$
		
By definition, as long as $\lra_{lra}^{(2)}$ is satisfied by its ancestors, $S_2$ can satisfy its guarantee under all admissible environment conditions. Note that $\lra_{lra}^{(2)}$ introduces an additional guarantee for the remaining subsystems, i.e. $S_1$.

Upon synthesizing a local controller for $S_2$, we update the global contract $\mathbb{C}' = [e_1,y_1].$

Now we delete $S_2$ from the system graph and continue with $S_1$. %New contract is already in local form but 
First, the local assumption and guarantee are computed:
\begin{align*}
\Asmp^{\downarrow(1)} &= \Asmp|_{V_{e^{ext}_1}} = e_1 
&\Gar^{\downarrow(2)} &= \lra_{lra}^{(2)} = y_1
\end{align*}		
Note that the local contract is satisfiable. Thus we can synthesize a local controller for $S_1$ by solving the respective QSAT problem. 

Now we provide three examples and show why we need the conditions on specifications and system graph.
To begin with, when the first condition in Theorem \ref{theo:comp} is satisfied, \eqref{eq:asmp} becomes an equality. 

In other words, Boolean projection does not result in information loss. When this condition is violated, some information is lost. This might lead to incomplete solutions as illustrated by the next example.

\begin{example}\label{ex:asmp}
	For the system graph seen Fig.~\ref{fig:ex1} (Left), let $\mathbb{C} = \left[e_1 \oplus e_2,y_2\right]$
	%{\footnote{The symbol $\oplus$ used for XOR operation}} 
	where 
	\begin{align*}
	f_1 &:(u_1,e_1) \mapsto (e_1 \land u_1)\\
	f_2 &: (u_1,e_2,y_1)\mapsto ((e_2 \lor y_1)\land u_2).
	\end{align*}
\end{example}

In this example, the assumption cannot be written as a conjunction of two local assumptions, hence the completeness conditions are violated. Let us apply the algorithm to see why this causes a problem. We start by finding local assumption-guarantee pair and the corresponding least restrictive assumption.
\begin{align*}
	A^{\downarrow(2)} &= A|_{V_{e_2^{ext}}} = True, &G^{\downarrow(2)} &= y_2,\\ G^{\uparrow(2)} &= True,&\lra_{lra}^{(2)} &= y_1.
	%\end{split}
	\end{align*}
Then, we update the contract as
$\mathbb{C}' = [e1 \oplus e_2, y_1]$.
We delete $S_2$ and move onto $S_1$. Unfortunately the local assumption $A^{\downarrow(1)} = A|_{V_{e_1^{ext}}} = True$ together with $G^{\downarrow(1)} = y_1$ is unsatisfiable.

On the other hand, the following local contracts leads to local controllers satisfying the global specifications: $\mathbb{C}_2 = [e_2 \lor y_1,y_2]$,
$\mathbb{C}_1 = [e_1,y_1].$

When the second condition in Theorem \ref{theo:comp} is violated, $
\llbracket\Gar^{\downarrow(i)}\rrbracket \times \ws{\Gar^{\uparrow(i)}}\subset \llbracket G\rrbracket
$ is a strict subset relation. This implies that Boolean distribution is conservative. Next example shows how this results in incomplete solutions.
\begin{example}\label{ex:gar}
	For the system graph seen Fig, let $\mathbb{C} = \left[ True, y_1 \lor y_2\right]$ where 
	$$f_1 = (u_1,e_1) \mapsto e_1 \land u_1$$ 
	$$f_2 = (u_2,e_2,y_1) \mapsto (e_2 \lor u_2) \land \neg y_1$$
\end{example}

In this example, second condition is violated. 
Then the distribution is not unique. We have
$$
\gamma_1 =\{\Gar^{\downarrow(2)}_1 =y_2, \Gar^{\uparrow(2)}_1 = True\}$$
or
$$\gamma_2= \{\Gar^{\downarrow(2)}_2 = True, \Gar^{\uparrow(2)}_2 = y_1\}. 
$$

Using $\Gar^{\uparrow(2)}_1$ results in $\lra^{(2)} = y_1$. Then the local contract $\mathbb{C}_1 = {True,y_1}$ is not satisfiable.
Similarly the second option also forces $S_1$ to set $y_1 =True$. Thus again, resulting in an unsatisfiable local contract.

However, the following local contracts can be used to synthesize local controllers satisfying $\mathbb{C}$: $\mathbb{C}_2 = [\neg y_1,y_2]$ and
$\mathbb{C}_1 = [e_1,y_1]$.

Finally, if a subsystem has multiple parents, the least restrictive assumption might lead to violation of the second condition. This is demonstrated by the next example.

\begin{example}\label{ex:forest}
	For the system graph seen Fig.~\ref{fig:ex1} (Right), let $\mathbb{C} = \left[ True, y_3 \right]$ where 
	\begin{align*}
	f_1 &: (u_1,e_1) \mapsto (e_1 \land u_1)\\
	f_2 &: (u_2,e_2,y_1) \mapsto (e_2 \lor u_2) \land \neg y_1\\
	f_3 &: (u_3,y_2,y_1) \mapsto (y_1 \lor y_2).
	\end{align*}
\end{example}

In this example, after synthesizing a controller for $S_3$, global contract is updated to $\mathbb{C}' = [True,y_1 \lor y_2].$

Note that the remaining problem is equal to Example \ref{ex:gar}, thus the solution is not complete.

Finally, it is worth mentioning that it is easy to construct examples where there exists a distributed controller but some of the conditions in Theorem \ref{theo:comp} fail to hold but the proposed algorithm still finds a valid solution.

%% file: 7_eps_desc.tex
%!TEX root = dist_vms.tex
In this section, 
we model electric power system of an aircraft as a Boolean network and and translate the specifications into an assumption-guarantee pair. Then we use Algorithm~\ref{alg:main} to synthesize distributed controllers.
Figure \ref{fig:largecircuit} (Left) shows a circuit for power generation and distribution in an aircraft in the form of a single-line diagram \cite{moir_aircraft_book}, a simplified notation for drawing three-phase power systems.
To model the system as a Boolean network, we first represent the electric power system topology by a digraph $\mathcal{G} = (\Node, \Edge)$ that indicates the power flow directions. The set $\Node$ of nodes  
in the graph consists of components such as generators%($\mathcal{E}$)
, rectifier units% ($\mathcal{R}$)
, buses% ($\mathcal{B}$)
, transformers % ($\mathcal{T}$) 
 and dummy nodes where a collection of wires meet. The set $\Edge$ of edges contains contactors %($\mathcal{C}$) 
 and solid wire links between other components. We locate the maximal connected components of this graph; and treat each connected component as a Boolean subsystem. This results in the interconnection structure shown in Fig.~\ref{fig:largecircuit} (Right).  
 
In order to analyze the steady-state behavior of the electric power system in different failure modes, we use a discrete abstraction. The status of contactors %$C\in\mathcal{C}$ 
can either be \emph{open} or \emph{closed}, i.e. Boolean $\{False,True\}$. Contactors are controllable (control inputs). Elements in the sets of generators %$\mathcal{E}$ 
and rectifier units %$\mathcal{R}$ 
and transformers %$\mathcal{T}$ 
are uncontrollable (environment inputs), and can take values of $True$ (i.e., the component is online and outputting the correct voltage), or $False$ (i.e., the component failed, no power output, open circuit). Power availability on a bus (outputs) depends on the status of generators, rectifier units, transformers and contactors. We say that there is a \emph{live path} between two components if there exists a simple path in the graph $\mathcal{G}$ that connects the two nodes corresponding to these components, there is no offline component along the path including end nodes, and the contactors along this path are all closed. The status of a bus $B$ %$B\in\mathcal{B}$ 
can be (i) $True$ (\emph{powered}): if there is a live path between $B$ and some generator $E$ %$E\in\mathcal{E}$ 
(not offline by definition of live path), (ii) $False$ (\emph{unpowered}): there is no live path between $B$ and any generator $E$.%$E\in \mathcal{E}$.

\begin{figure}[ht]
	\centering
	\subfigure{\includegraphics [width=.5\textwidth] {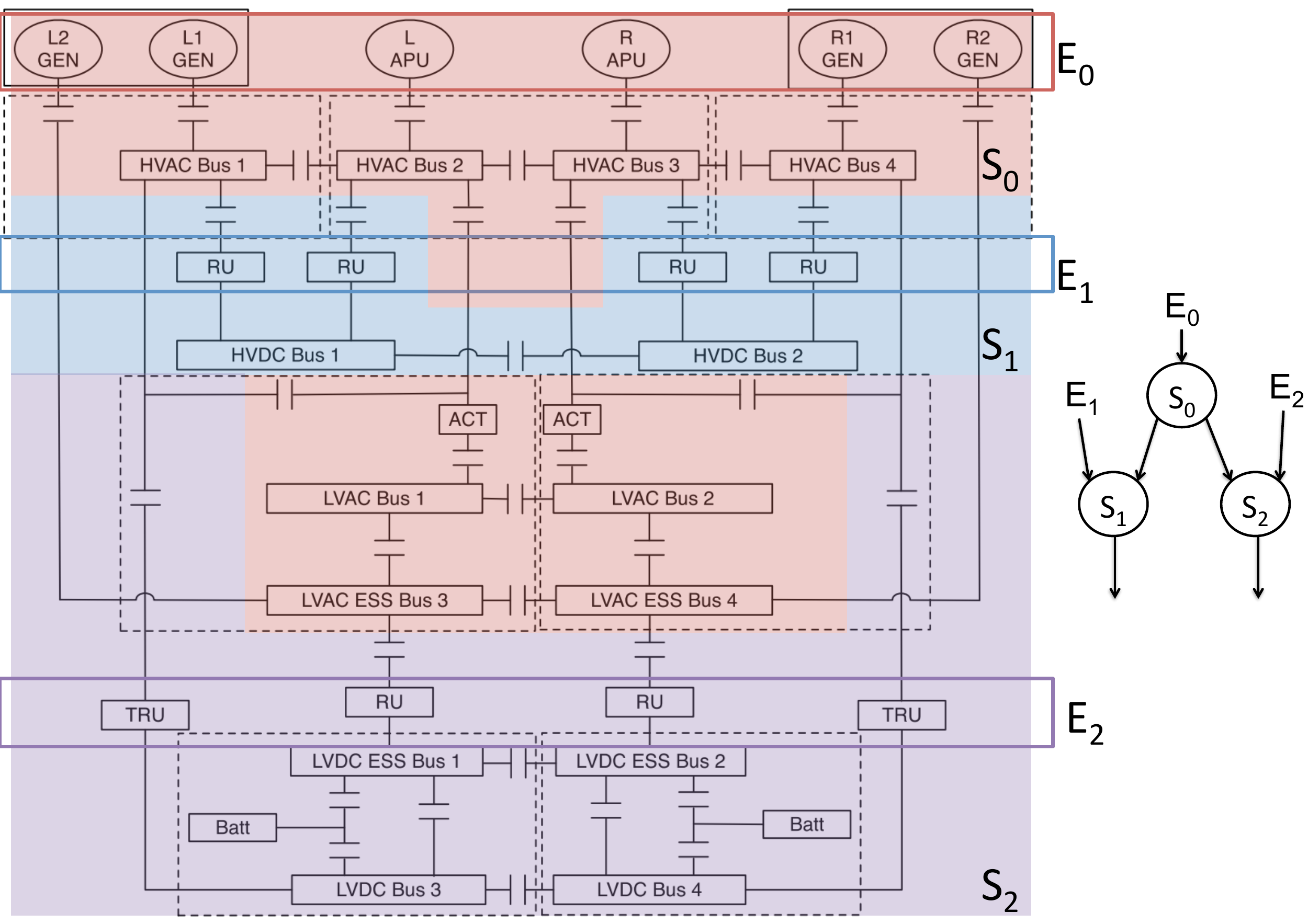}}
	\caption{\small (Left): Single-line diagram of the aircraft EPS (adapted from \cite{nuzzocontract}). Different colors indicate different subsystems identified. (Right): Corresponding Boolean network.}
	\label{fig:largecircuit}
\end{figure}

Boolean network representing the circuit has 22 environment inputs, 32 control inputs and 15 outputs in total. The global assumption, that at least one power source is healthy along with one rectifier from each $E_1$ and $E_2$, has 22 variables and 14 conjunctive clauses. We require all buses to be powered and no AC coupling between power sources, that is there should not be a live path between any AC power source. Resulting  the global guarantee has 47 variables and 2142 conjunctive clauses. Solving this problem in a centralized manner requires solving a QSAT problem with 69 variables and 2151 clauses. The computation takes 336 seconds using PicoSAT \cite{Biere_picosatessentials} on a laptop with 2.5 GHz Intel Core i7 processor and 16 GB of RAM. Keep in mind that the returned control protocol is a centralized one and is not implementable in a distributed manner in general.

Then we use Alg.~\ref{alg:main} to solve for the same case. Now we have three Boolean subsystems. This means we have to solve three QSAT problems that are much smaller with 39, 19, and 24 variables and 485, 369, and 1294 clauses (for $S_0,S_1$ and $S_2$, respectively). The system graph is a tree (hence satisfies the forest condition) and the specification satisfies \eqref{eq:conj}. This implies that Alg.~\ref{alg:main} is complete. The computation takes 42 seconds on the same laptop and returns a distributed controller.

%% file: 8_conclusion.tex
%!TEX root = dist_vms.tex

In this paper, we proposed an algorithm for synthesizing local specifications and corresponding local controllers for a Boolean network to guarantee the satisfaction of a global specification. The algorithm is shown to be sound, and it is complete when the global specification and networks interconnection structure satisfy certain structural properties. An application of the proposed approach is demonstrated by synthesizing distributed controllers for an aircraft electric power system. 

In the future, we will extend the presented algorithms to dynamic discrete transition systems and specifications given in certain fragments of linear temporal logic. It would also be interesting to understand the connections between the structural properties identified in this work with those that are known in decentralized control for different system classes and control objectives. %(such as \cite{witsenhausen1971separation,rotkowitz2006characterization,nayyar2014optimal,FinkbeinerS05}). 

%Hardness of synthesizing distributed controllers is known in many different domains, including synthesis from temporal logic constraints \cite{pnueli90}, stochastic and optimal control \cite{witsenhausen1968counterexample}. There are also several results in these domains that show that information structure of the controllers, and structure of the control objectives play a role on the hardness of the problem \cite{witsenhausen1971separation,yoshikawa1978decomposition,rotkowitz2006characterization,nayyar2014optimal,FinkbeinerS05}. It would also be interesting to understand the connect

%Our future research will focus on 

%We proposed a simple, intuitive algorithm for co-design of distributed control architectures and synthesis and implementation of control protocols on such architectures for aircraft electric power systems. 

%We are concerned with logical consistency and distributed synthesis. Each synthesized control protocol itself can be implemented in a distributed manner (software redundancy).

%Separation in time scales of voltages and currents and reconfigurations enables reasoning about the reconfiguration logic at a higher abstraction level independent of underlying physical dynamics.

%Guarantees sep + tree: no merging,  assumptions sep or information structure: related to signaling 

%% file: 9_appendix.tex
%!TEX root = dist_vms.tex

\subsection*{Computation of local guarantees}

Boolean distribution operation is done by representing the Boolean formulas as graphs and using graph properties. This has two advantages: (i) graphs provide a canonical representation for the satisfying set of the formulas (whereas there could be multiple formulas with the same satisfying set), (ii) once the problem is converted to a graph problem, we leverage well-established algorithms from graph theory to find distributions.  %Let $\mathcal{H} = (\Node, \Edge)$ be a graph.

Let us start with some additional graph terminology. Given a graph, a set of nodes, no two of which are connected with an edge is called an \emph{independent set}. 
A graph $\mathcal{H} =(\Node, \Edge)$ is called \emph{bipartite} if its nodes can be partitioned into two independent sets, i.e. $\Node = \Node_1 \cup \Node_2$ and $\Edge \subset  \Node_1 \times \Node_2$.  
$\Node_1$ and $\Node_2$ are also called \emph{parts}.
If every node in one part is connected to every other node in the other part with an edge, the bipartite graph is called \emph{complete}, i.e., $\Edge =  \Node_1 \times \Node_2$. %$\forall u \in \Node_1:\forall v \in \Node_2: (u,v) \in \Edge$.

Let $H_c=(\Node',\Edge')$ be a complete bipartite subgraph of a bipartite graph $H=(\Node,\Edge)$. Then $H_c$ is called \emph{maximal} if addition of any node to $\Node'$ breaks down the completeness. Putting it differently, $H_c$ is not a subgraph of any other $H'_c$ which is also a complete bipartite subgraph of $H$.

Let $S$ be a Boolean system and $\mathbb{C}=[\Asmp,\Gar]$ be its specification. Now assume that $S_i$ is a leaf node in the system graph and we want to compute its local guarantee. %To perform Boolean distribution, we use a bipartite graph.
In order to perform Boolean distribution, we construct a bipartite graph $H= \left( \left(\prod_{j\neq i} V_{y_j} \bigcup V_{y_i}\right), \Edge \right)$. In this graph, every node $y_i \in V_{y_i}$ represents a different valuation of $S_i$'s outputs. Similarly each node $y \in \prod_{j\neq i} V_{y_j}$ represent all other outputs valuations. Edges are created between nodes that forms an admissible output, i.e.,
$\Edge = \{ (y_i,y) \mid (y_i,y) \in \ws{\Gar} \}.$
 %$\Edge = \{ (y_i,y) \mid y_i \in V_{y_i}, y_j \in \prod_{j\neq i} V_{y_j} \text{ such that } (y_i,y) \in \ws{\Gar} \}.$

\begin{prop}\label{prop:biclique}
	Let $H_c = \left( \left( \prod_{j\neq i} \bar{V}_{y_j}\bigcup\bar{V}_{y_i} \right), \bar{\Edge} \right) $ be a subgraph of $H$ where $\bar{V}_{y_k} \subset {V}_{y_k}$ for all $k$. 
	Now define $\Gar^{\downarrow(i)} : V_{y_i} \to \mathbb{B}$ and $\Gar^{\uparrow(i)} : \prod_{j\neq i} V_{y_j} \to \mathbb{B}$ with their satisfying set $\ws{\Gar^{\downarrow(i)}} = \bar{V}_{y_i}$ and $\ws{\Gar^{\uparrow(i)}} = \prod_{j\neq i} \bar{V}_{y_j}$. Then $\gamma=\{\Gar^{\downarrow(i)},\Gar^{\uparrow(i)}\}$ is a maximal distribution if and only if $H_c$ is a maximal complete bipartite graph.
\end{prop} 

\begin{proof}
	Let $\gamma$ be a distribution. % By construction of $H$, $(y_i,y) \in \Edge$. 
	Now define $H_c = \left( \left( \prod_{j\neq i} \bar{V}_{y_j}\bigcup\bar{V}_{y_i} \right), \bar{\Edge} \right)$ where $ \bar{V}_{y_i}=\ws{\Gar^{\downarrow(i)}},  \prod_{j\neq i} \bar{V}_{y_j}= \ws{\Gar^{\uparrow(i)}}$ and $\bar{\Edge} = \Edge \cap \left( \prod_{j\neq i} \bar{V}_{y_j}\times\bar{V}_{y_i} \right)$. By definition of distribution $\forall y_i \in \ws{\Gar^{\downarrow(i)}} : \forall y \in \ws{\Gar^{\uparrow(i)}}: (y_i,y) \in \ws{\Gar}$. This implies $\forall y_i \in \bar{V}_{y_i} : \forall y \in \prod_{j\neq i} \bar{V}_{y_j}: (y_i,y) \in \bar{\Edge}$. Thus $H_c$ is a complete bipartite graph.
	
	Conversely let $H_c = \left( \left( \prod_{j\neq i} \bar{V}_{y_j}\bigcup\bar{V}_{y_i} \right), \bar{\Edge} \right) $ be a complete bipartite subgraph of $H$.  By definition of complete bipartite graph $\forall y_i \in \bar{V}_{y_i} : \forall y \in \prod_{j\neq i} \bar{V}_{y_j}: (y_i,y) \in \bar{\Edge} \subset \Edge$.  Now define $\Gar^{\downarrow(i)} : V_{y_i} \to \mathbb{B}$ and $\Gar^{\uparrow(i)} : \prod_{j\neq i} V_{y_j} \to \mathbb{B}$ with their satisfying set $\ws{\Gar^{\downarrow(i)}} = \bar{V}_{y_i}$ and $\ws{\Gar^{\uparrow(i)}} = \prod_{j\neq i} \bar{V}_{y_j}$. Then $\forall y_i \in \ws{\Gar^{\downarrow(i)}} : \forall y \in \ws{\Gar^{\uparrow(i)}}: (y_i,y) \in \ws{\Gar}$ by construction of $H$. This implies $\gamma=\{\Gar^{\downarrow(i)},\Gar^{\uparrow(i)}\}$ is a local guarantee set.
	
	Finally let the distribution $\gamma=\{\Gar^{\downarrow(i)},\Gar^{\uparrow(i)}\}$ be maximal. This implies that there does not exist another complete bipartite subgraph $\bar{H_c}=\{\prod_{j\neq i} \bar{V}_{y_j}\cup\bar{V}_{y_i}, \bar{\Edge} \}$ such that $\prod_{j\neq i} \bar{V}_{y_j} \supset \prod_{j\neq i} {V}_{y_j}$ and $\bar{V}_{y_i} \supset {V}_{y_i}$. This implies that there does not exist any other subgraph $H'_c$ of $H$ such that $H_c$ is a subgraph of $H'_c$. Thus $H_c$ is maximal. Going in the other direction is also similar. If $H_c$ is maximal, so is $\gamma=\{\Gar^{\downarrow(i)},\Gar^{\uparrow(i)}\}$.
	\end{proof}

Note that with Proposition~\ref{prop:biclique},  the problem of finding a distribution reduces to the problem of finding maximal complete bipartite subgraphs. The latter can be done using standard results from graph theory literature \cite{alexe2004consensus}. In particular, in our implementation, we have used a variant of the Bron-Kerbosch algorithm \cite{bron1973algorithm} to find all maximal distributions.

%% file: dist_vms.bbl
\begin{thebibliography}{10}

\bibitem{alexe2004consensus}
G.~Alexe, S.~Alexe, Y.~Crama, S.~Foldes, P.~L. Hammer, and B.~Simeone.
\newblock Consensus algorithms for the generation of all maximal bicliques.
\newblock {\em Discrete Applied Mathematics}, 145(1):11--21, 2004.

\bibitem{Benvenuti08}
L.~Benvenuti, A.~Ferrari, E.~Mazzi, and A.~L. Vincentelli.
\newblock Contract-based design for computation and verification of a
  closed-loop hybrid system.
\newblock In {\em Proceedings of the Hybrid Systems: Computation and Control},
  pages 58--71, 2008.

\bibitem{Biere_picosatessentials}
A.~Biere.
\newblock Pico{SAT} essentials.
\newblock {\em Journal on Satisfiability, Boolean Modeling and Computation
  (JSAT)}, 2008.

\bibitem{bozzano2014formal}
M.~Bozzano, A.~Cimatti, C.~Mattarei, and S.~Tonetta.
\newblock Formal safety assessment via contract-based design.
\newblock In {\em Automated Technology for Verification and Analysis}, pages
  81--97. Springer, 2014.

\bibitem{bron1973algorithm}
C.~Bron and J.~Kerbosch.
\newblock Algorithm 457: finding all cliques of an undirected graph.
\newblock {\em Communications of the ACM}, 16(9):575--577, 1973.

\bibitem{chatterjee2013distributed}
K.~Chatterjee, T.~A. Henzinger, J.~Otop, and A.~Pavlogiannis.
\newblock Distributed synthesis for ltl fragments.
\newblock In {\em Formal Methods in Computer-Aided Design (FMCAD), 2013}, pages
  18--25. IEEE, 2013.

\bibitem{cheng2010analysis}
D.~Cheng, H.~Qi, and Z.~Li.
\newblock {\em Analysis and control of Boolean networks: a semi-tensor product
  approach}.
\newblock Springer Science \& Business Media, 2010.

\bibitem{de2001interface}
L.~De~Alfaro and T.~A. Henzinger.
\newblock Interface theories for component-based design.
\newblock In {\em Embedded Software}, pages 148--165. Springer, 2001.

\bibitem{FinkbeinerS05}
B.~Finkbeiner and S.~Schewe.
\newblock Uniform distributed synthesis.
\newblock {\em Logic in Computer Science, Symposium on}, 0:321--330, 2005.

\bibitem{hearn2009games}
R.~A. Hearn and E.~D. Demaine.
\newblock {\em Games, puzzles, and computation}.
\newblock CRC Press, 2009.

\bibitem{KAUFFMAN1969437}
S.~Kauffman.
\newblock Metabolic stability and epigenesis in randomly constructed genetic
  nets.
\newblock {\em Journal of Theoretical Biology}, 22(3):437 -- 467, 1969.

\bibitem{madhusudan01}
P.~Madhusudan and P.~Thiagarajan.
\newblock Distributed controller synthesis for local specifications.
\newblock In F.~Orejas, P.~Spirakis, and J.~van Leeuwen, editors, {\em
  Automata, Languages and Programming}, volume 2076 of {\em Lecture Notes in
  Computer Science}, pages 396--407. Springer Berlin, 2001.

\bibitem{moir_aircraft_book}
I.~Moir and A.~Seabridge.
\newblock {\em Aircraft Systems: Mechanical, Electrical, and Avionics
  Subsystems Integration}.
\newblock AIAA Education Series, 2001.

\bibitem{nuzzocontract}
P.~Nuzzo, H.~Xu, N.~Ozay, J.~B. Finn, A.~L. Sangiovanni-Vincentelli, R.~M.
  Murray, A.~Donz{\'e}, S.~Seshia, et~al.
\newblock A contract-based methodology for aircraft electric power system
  design.
\newblock {\em Access, IEEE}, 2:1--25, 2014.

\bibitem{papadimitriou2003computational}
C.~M. Papadimitriou.
\newblock {\em Computational complexity}.
\newblock Addison-Wesley, 1994.

\bibitem{pnueli90}
A.~Pneuli and R.~Rosner.
\newblock Distributed reactive systems are hard to synthesize.
\newblock In {\em SFCS '90: Proceedings of the 31st Annual Symposium on
  Foundations of Computer Science}, pages 746--757 vol.2, Washington, DC, USA,
  1990. IEEE Computer Society.

\bibitem{Sahin16MSC}
Y.~E. Sahin and N.~Ozay.
\newblock {SAT}-based distributed reactive control protocol synthesis for
  boolean networks.
\newblock In {\em IEEE Multi-Conference on Systems and Control (MSC)}, 2016.

\bibitem{sahin16}
Y.~E. Sahin and N.~Ozay.
\newblock {WiP} abstract: Distributed reactive control synthesis for aircraft
  electric power systems via {SAT} solving.
\newblock In {\em ACM/IEEE International Conference on Cyber-Physical Systems
  (ICCPS)}, 2016.

\bibitem{tabuada2009verification}
P.~Tabuada.
\newblock {\em Verification and control of hybrid systems: a symbolic
  approach}.
\newblock Springer, 2009.

\end{thebibliography}
